\documentclass[conference,a4paper]{IEEEtran}
\usepackage{amssymb,amsmath}
\usepackage{cite}
\usepackage{graphicx,subfigure}
\usepackage{psfrag}
\usepackage{url}
\usepackage[latin1]{inputenc}
\usepackage[absolute,overlay]{textpos}
\usepackage{tikz}
\usetikzlibrary{arrows,calc,decorations.markings}
\usetikzlibrary{topaths}
\usetikzlibrary{shapes,trees,snakes}
\usetikzlibrary{arrows}
\usetikzlibrary{shadows}
\usetikzlibrary{positioning}
\usetikzlibrary{matrix}
\usetikzlibrary{shapes.geometric}
\usetikzlibrary{decorations.pathmorphing}
\usepgflibrary{patterns}
\usetikzlibrary{calc}
\usetikzlibrary{fit}					
\usetikzlibrary{backgrounds}	
\newcommand{\argmax}[1]{{\operatorname{arg}\,\max_{#1}}\,}

\usepackage{pgf}
\usetikzlibrary{arrows,automata}
\usepackage[latin1]{inputenc}

\usepackage{algorithm}
\usepackage{algpseudocode}

\usepackage{amsthm}

\newtheorem{theorem}{Theorem}
\newtheorem{corollary}{Corollary}
\newtheorem{remark}{Remark}

\begin{document}

\title{Rate Control under Finite Blocklength for Downlink Cellular Networks with Reliability Constraints}
\author{
	\IEEEauthorblockN{	Onel L. Alcaraz L\'{o}pez, %
		Hirley Alves, 
		Matti Latva-aho
	}	
	\IEEEauthorblockA{Centre for Wireless 	Communications (CWC), Oulu, Finland\\}
	\{onel.alcarazlopez, hirley.alves, matti.latva-aho\}@oulu.fi}	
\maketitle

\begin{abstract}
	Coming cellular systems are envisioned to open up to new  services and applications with high reliability and low latency requirements. In this paper we focus on the rate allocation problem in downlink cellular networks with Rayleigh fading and stringent reliability constraints. We propose a rate control strategy to cope with those requirements making use only of topological characteristics of the scenario, the reliability constraint and the number of antennas that are available at the receiver side. Numerical results show the feasibility of the ultra-reliable operation when the number of antennas increases, and also that our results remain valid even when operating at short blocklength as far as the amount of information to be transmitted is not too small.
\end{abstract}
\section{Introduction}
The advent of fifth generation (5G) of wireless systems opens up new possibilities and gives rise to a new category of use cases termed ultra reliable low latency communications (URLLC) \cite{Popovski.2014}, where services are characterized by very stringent requirements. Some examples are \cite{Schulz.2017}:
 factory automation, with maximum latency around $0.25$-$10$ms and maximum error probability
 of $10^{-9}$; smart grids ($3$-$20$ms, $10^{-6}$), professional audio ($2$ms, $10^{-6}$), etc. In general, there is a fundamental trade-off between delay and reliability metrics due to the fact that by relaxing one of them, we can enhance the performance of the other. In fact, Long-Term Evolution (LTE) already offers guaranteed bit rate that can support packet error rates down to $10^{-6}$, however, the delay budget goes up to
 $300$ms which includes radio, transport and core network latencies \cite{Singh.2017}. Therefore, the interplay between these metrics makes physical layer design of URLLC very complicated \cite{Hyoungju.2017}.

 The principles for supporting URLLC are discussed in \cite{Popovski.2017} including various elements of the system design, such as use of various  diversity sources, design of packets and access protocols. Shared diversity resources are explored more deeply in \cite{Kotaba.2018}  when multiple  connections are only intermittently active in order to support URLLC. Authors in \cite{She.2017} address the 
 delay and packet loss components in URLLC and the network availability for supporting the quality of service of users, while some tools for resource optimization are presented. 
 The minimum energy required to transmit $k$ information bits with a given reliability over a multiple-antenna Rayleigh block-fading channel is investigated in  \cite{Yang.2017}, while also in a multi-antenna setup the trade-off between reliability, throughput, and latency when transmitting short packets is identified in  \cite{Durisi.2016}.  
 In \cite{Dosti.2017}, an energy efficient power allocation strategy for the Chase Combining Hybrid Automatic Repeat Request (CCHARQ) is proposed to meet the reliability constraints of URLLC systems. Cooperative communications are also considered in literature, e.g., \cite{Nouri.2017}, and \cite{Lopez.2017,Lopez2.2017} for wireless powered communications, as a viable alternative to direct communication setups  \cite{Lopez3.2017}.
 
 In this paper we focus on the rate allocation problem in downlink URLLC cellular networks with Rayleigh fading. The system is composed of a multi-cell setup where multiple base stations (BSs) are interfering an URLLC link with multiple antennas at receiver side. The main contributions of this work can be listed as follows:
 \begin{itemize}
 	\item we propose a rate allocation scheme that meets the stringent reliability constraints of the system. The allocated rate depends only on topological characteristics of the scenario, the reliability constraint and the number of antennas that are available at the user equipment device (UE) side;
 	\item  we attain accurate closed-form approximations for the rate to be allocated when the UE operates using the Selection Combining (SC) and Maximum Ratio Combining (MRC) schemes;
 	\item numerical results show the superiority of the MRC scheme and also the feasibility of the ultra-reliable operation when the number of antennas increases at the UE. 
 	\item we show that our analytical results remain valid even when operating at short blocklength, low latency, as far as the amount of information to be transmitted is not too small.
 \end{itemize}

Next, Section \ref{system} introduces the system model and assumptions. Section \ref{INF} presents the rate allocation strategy, while Section \ref{FIN} discuss some aspects related with the low latency requirement. Finally, Section \ref{conclusions} concludes the paper.

\textit{Notation:} $X\!\sim\!\mathrm{Exp}(1)$ is a normalized exponential distributed random variable with Cumulative Distribution Function (CDF)  $F_X(x)=1-e^{-x}$, while $Y\!\sim\! L(p,q)$ is a Lomax random variable with Probability Density Function (PDF) $f_Y(y|p,q)\!=\!q\big(1\!+\!\frac{q}{p}y\big)^{-1\!-\!p}\!$ and  CDF $F_Y(y|p,q)\!=\!1\!-\!\big(1\!+\!\frac{q}{p}y\big)^{-p}$. Also, $Q(x)\!=\!\int_{x}^{\infty}\!\!\frac{1}{\sqrt{2\pi}}e^{-t^2/2}\mathrm{d}t$ and $\Gamma(p,x)=\int_x^{\infty}t^{p-1}e^{-t}\mathrm{d}t$ are the Gaussian Q-function and the incomplete gamma function, respectively.
\section{System Model}\label{system}
Consider a multi-cell downlink cellular network where a collection of $\eta+1$ BSs, $\mathrm{BS}_0, \mathrm{BS}_1,..., \mathrm{BS}_{\eta}$,  are spatially distributed in a given area $\mathcal{A}\subseteq\mathbb{R}^2$. We denote the collection of those points as $\Phi$, and this deployment can be seen in general as an instantaneous realization of some Point Process (PP). We focus on the transmission over one given channel while assuming that each $\mathrm{BS}_j, j=0,...,\eta$, is currently using that channel to transmit data to its corresponding $\mathrm{
UE}$, $\mathrm{UE}_j$. We denote the distance between each $\mathrm{BS}_j$ and $\mathrm{UE}_0$ as $r_j$ and adopt a channel model that comprises standard path-loss with exponent $\alpha$ and Rayleigh fading. We define the link between $\mathrm{BS}_0$ and $\mathrm{UE}_0$, as the \textit{typical link}, and we focus on its performance. Fig.~\ref{Fig1} shows an example topology with $\eta=10$ interfering BSs.

$\mathrm{UE}_0$ is equipped with $M$ antennas sufficiently separated such that the fading affecting the received signal in each antenna can be assumed independent and full gain from spatial diversity can be attained. Channel State Information (CSI) is available at $\mathrm{UE}_0$, and each BS transmits with fixed power. We consider an interference-limited wireless system given a dense deployment of small cells where the impact of noise is neglected\footnote{However, the impact of the noise could easily be incorporated without substantial changes.}. Thus, the Signal-to-Interference Ratio (SIR) perceived in each antenna is 
%
$\mathrm{SIR}_i={h_ir_0^{-\alpha}}/{I_i} 
$, 
with
\vspace*{-2mm}
\begin{align}
I_i=\sum_{j\in\Phi^*} g_{j,i}r_{j}^{-\alpha},\label{Ii}
\vspace*{-1mm}
\end{align}
where $h_i,g_{j,i}\sim\mathrm{Exp}(1)$ are the power channel gain coefficients of the typical link at the $i$-th antenna, and of the interfering channel between $j$-th BS and $i$-th antenna of $\mathrm{UE}_0$, respectively. Finally, consider that signal transmitted by $\mathrm{BS}_0$ spans over $n$ channel uses and containing  $k$ information bits, e.g., fixed transmission rate $r=k/n$ (bpcu).
\begin{figure}[t!]
	\centering
	\subfigure{\includegraphics[width=0.42\textwidth]{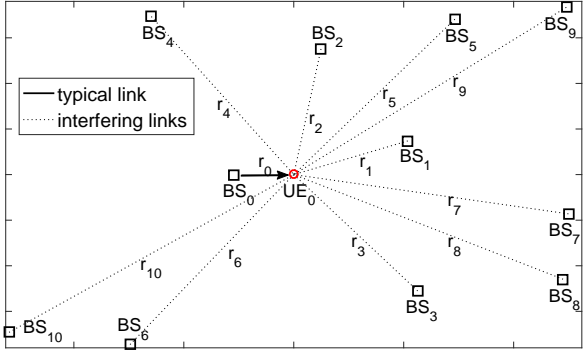}}
	\vspace*{-2mm}
	\caption{Illustration of the system model with $\eta=10$.}
	\vspace*{-2mm}		
	\label{Fig1}
\end{figure}
\section{Asymptotic rate control under reliability constraint}\label{INF}
Herein we consider the asymptotic formulation (infinite blocklength), in which an error in decoding the information occurs when $\mathrm{SIR}<\theta$, $\theta=2^{r}-1$. Therefore, the distribution of the SIR, $\mathbb{P}(\mathrm{SIR}<\theta|\Phi)=F_{\mathrm{SIR}_i}(\theta|\Phi)$, is necessary and it is given in the following result.
\begin{theorem}\label{the1}
	The CDF of the SIR at each antenna $i=1,...,M$ is given by
	\begin{align}
	F_{\mathrm{SIR}_i}(\gamma|\Phi)=1-\prod_{j\in\Phi^*}\frac{1}{1+\gamma r_0^{\alpha}r_{j}^{-\alpha}},\label{Fsnr}
	\end{align}	
	which is upper-bounded by
	\begin{align}
	F_{\mathrm{SIR}_i}(\gamma|\Phi)\approx 1-\Big(1+\frac{\gamma}{\eta}\beta\Big)^{-\eta}  \label{Fsnrapp} 
	\end{align}
	with  $\beta=r_0^{\alpha}\sum_{j=1}^{\eta}r_{j}^{-\alpha}$ and $\Phi^*=\Phi\backslash \mathrm{BS}_0$. 
\end{theorem}
\begin{proof}
	We proceed as follows
	\begin{align}
	F_{\mathrm{SIR}_i|\Phi}(\gamma)&=\mathbb{P}\big(\mathrm{SIR}_i<\gamma\big|\Phi\big)=1-\mathbb{P}\big(\mathrm{SIR}_i>\gamma\big|\Phi\big)\nonumber\\
	&=1-\mathbb{P}\big(h_i>\gamma r_0^{\alpha} I_i\big|\Phi\big)\stackrel{(a)}{=}1-\mathbb{E}_g \big[ e^{-\gamma r_0^ {\alpha} I_i}\big]\nonumber\\
	&=1-\prod_{j\in\Phi^*}\mathbb{E}_g \big[e^{-\gamma g_{j,i}r_0^{\alpha}r_{j}^{-\alpha}}\big], \label{FsnrP}
	\end{align}
	where $(a)$ follows from the complementary CDF of exponential random variable $h_i$, and \eqref{Fsnr} comes directly after \eqref{FsnrP}. Now we focus on the upper bound.
	\begin{align}
	\prod_{j\in\Phi^*}(1\!+\!\gamma r_0^{\alpha}r_{j}^{-\alpha})&\!=\!\prod_{j=1}^{\eta}(1\!+\!\gamma r_0^{\alpha}r_{j}^{-\alpha})\!\stackrel{(a)}{\le}\!\Bigg[\frac{\sum\limits_{j=1}^{\eta}\big(1\!+\!\gamma r_0^{\alpha}r_{j}^{-\alpha}\big)}{\eta}\Bigg]^{\eta}\nonumber\\
	&\!\stackrel{(b)}{\le}\!\bigg[1\!+\!\frac{\gamma r_0^ {\alpha}}{\eta}\sum\limits_{j=1}^{\eta} r_{j}^{-\alpha}\bigg]^{\eta}\!\stackrel{(c)}{=}\!\bigg[1\!+\!\frac{\gamma}{\eta}\beta\bigg]^{\eta},\label{ap2}
	\end{align}
	where $(a)$ comes from  using the  relation between the geometric and the arithmetic mean, $(b)$ follows from simple algebraic transformations, and $(c)$ by adopting $\beta=r_0^{\alpha}\sum_{j=1}^{\eta} r_{j}^{-\alpha}$. Substituting \eqref{ap2} into \eqref{Fsnr} we attain \eqref{Fsnrapp}.	
\end{proof}
\begin{remark}\label{Re1}
	Both, \eqref{Fsnr} and \eqref{Fsnrapp}, converge in the left tail. This becomes evident from the proof of Theorem~\ref{the1}. Therein notice that when operating in the left tail $\prod_{j=1}^{\eta}(1+\gamma r_0^{\alpha}r_{j}^{-\alpha})$ should be close to 1, therefore each of the terms $(1+\gamma r_0^{\alpha}r_{j}^{-\alpha})\ge 1$ is expected to approximate to the unity. Hence, all of these terms are very similar between each other, and geometric mean approximates heavily  to arithmetic mean  in such scenarios. 
	\begin{figure}[t!]
		\centering
		\subfigure{\includegraphics[width=0.45\textwidth]{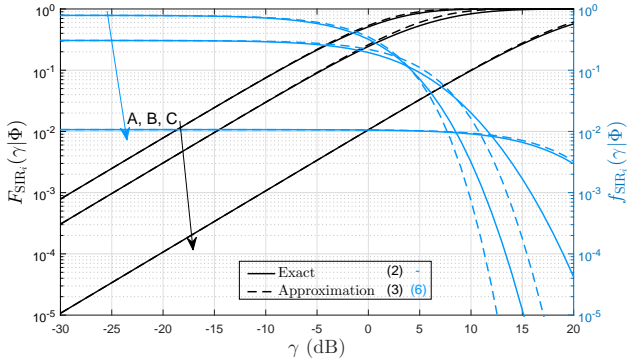}}
		\vspace*{-2mm}
		\caption{Comparison between the exact and approximate expressions of $F_{\mathrm{SIR}_i}(\gamma|\Phi)$ and $f_{\mathrm{SIR}_i}(\gamma|\Phi)$ for three different setups: $A:\ r_0=30$, $r_j=30+10j$, $j=1,...,20$; $B:\ r_0=20$, $r_j=10+20j$, $j=1,...,10$; $C:\ r_0=10$, $r_j=20+20j$, $j=1,...,4$. We set $\alpha=3.5$. The exact PDF was evaluated taking the derivative of the exact CDF \eqref{Fsnr} for each setup. }	
		\vspace*{-2mm}	
		\label{Fig2pp}
	\end{figure}	
	Also, as a consequence of \eqref{Fsnrapp} the SIR at each antenna $i=1,...,M$ is approximately a Lomax random variable with PDF given by
	\begin{align}
	f_{\mathrm{SIR}_i}(\gamma|\Phi)&\approx \beta\Big(1+\frac{\gamma}{\eta}\beta\Big)^{-\eta-1}.\label{pdf0}
	\end{align}		
    This can be represented as a scaled Lomax distribution such that $\mathrm{SIR}_i\approx\frac{\eta}{\beta} \varphi_i$ with $\varphi_i\sim L(\eta,1)$.
    
    The accuracy of the approximations in the left tail is clearly illustrated in Fig.~\ref{Fig2pp} for three different setups, thus, validating our findings. 
\end{remark}
\begin{remark}\label{Re0}
	Obtaining the PDF of the SIR directly from \eqref{Fsnr} seems intractable for large $\eta$, which is the case in dense network deployments. Also, since the upper bound is extremely tight in the left tail of the distribution, its utility is enormous because it is in that region where typical reliability constraints are, e.g., $\epsilon<10^{-1}$.
\end{remark}
\begin{remark}
	Notice also that our results are not restrictive to the adopted path loss model. In fact, $\beta$ can be expressed in a generalized way as $\beta=\ell_0\sum_{j=1}^{\eta}\ell_j$, where $\ell_0$ is the path loss in the typical link and $\ell_j,\ j=1,..,\eta$ is the path loss from the the $j$-th interfering BS to $\mathrm{UE}_0$. For static BS deployments these parameters can be easily obtained beforehand.
\end{remark}

Henceforth, we assume fixed $n$ and we are interested in finding the maximum number of bits, $k^*$, such that the reliability requirement is met. The target reliability is denoted by $1-\epsilon_{\mathrm{th}}$ where $\epsilon_{\mathrm{th}}$ is the admissible average error probability. In the following two subsection we solve this problem for SC and MRC schemes at the receiver side.
\subsection{Selection Combining -- SC}
When $\mathrm{UE}_0$ uses the SC scheme, the error in decoding the receiving message is defined as
\begin{align}
\epsilon&=\mathbb{P}\Big(\max_{i=1,...,M}\mathrm{SIR}_i<\theta\Big|\Phi\Big)=F_{\Omega}(\theta)\nonumber\\
&=\mathbb{P}\Big(\mathrm{SIR}_1<\theta,\mathrm{SIR}_2<\theta,\cdots,  \mathrm{SIR}_M<\theta \Big|\Phi\Big)\nonumber\\
&\stackrel{(a)}{=}\mathbb{P}\Big(\mathrm{SIR}_i<\theta\Big)^M\stackrel{(b)}{=}F_{\mathrm{SIR}_i}(\theta)^M,\label{errSC}
\end{align}
where $\Omega=\max\limits_{i=1,...,M}\mathrm{SIR}_i$, $(a)$ follows from the fact that $\mathrm{SIR}_i$ is distributed independently on each antenna. This is because the fading $h,g$ is i.i.d and the  topology is deterministic (non random). Finally, $(b)$ comes from using the definition of the CDF of $\mathrm{SIR}_i$.

\begin{theorem}
Assuming the distance from $\mathrm{UE}_0$ to the serving and interfering BSs is known and $\mathrm{UE}_0$ uses SC, the maximum number of bits to be transmitted, $k^*$,  while guaranteeing the reliability constraint given by $\epsilon_{\mathrm{th}}$, is the solution of%
\begin{align}
\prod_{j=1}^{\eta}\Big(1+(2^{k^*/n}-1) r_0^{\alpha}r_{j}^{-\alpha}\Big)=\frac{1}{1-\epsilon_{\mathrm{th}}^{1/M}},\label{eqSC}
\end{align}		
and is approximated by
\begin{align}
k^*\approx n\log_2\bigg(\frac{\eta}{\beta}\Big(\big(1-\epsilon_{\mathrm{th}}^{1/M}\big)^{-\frac{1}{\eta}}-1\Big)+1
\bigg).\label{eqSCap}
\end{align}
\end{theorem}
\begin{proof}
	Based on \eqref{errSC} we know that $k^*$ is the solution of $F_{\mathrm{SIR}_i}(2^{k/n}-1)^M=\epsilon_{\mathrm{th}}$. Using the exact expression for $F_{\mathrm{SIR}_i}(\theta)$ given in \eqref{Fsnr} we attain \eqref{eqSC}, while using \eqref{Fsnrapp} as an approximation of \eqref{Fsnr} we reach \eqref{eqSCap}.
\end{proof}
Notice that finding the solution of  \eqref{eqSC} for large $\eta$ is analytically heavy, and numerical methods would be required. Even for finite, and not so  large $\eta$, solving \eqref{eqSC} is difficult, thus, \eqref{eqSCap} provides an easy way of doing so. Also, since \eqref{Fsnrapp} is tight in the left tail, it is expected that \eqref{eqSCap} to be very accurate for $\epsilon_{\mathrm{th}}<10^{-1}$ and we provide numerical evidence in Section~\ref{results}.

Finally, the PDF of the SIR after SC is given by
\begin{align}
f_{\Omega}(x)&=\frac{d}{d x}F_{\Omega}(x)=\frac{d}{d x}F_{\mathrm{SIR}_i}(x)^M\nonumber\\
&=MF_{\mathrm{SIR}_i}(x)^{M-1}f_{\mathrm{SIR}_i}(x)\nonumber\\
&\approx M\beta\bigg(1\!-\!\Big(1\!+\!\frac{x}{\eta}\beta\Big)^{-\eta}\bigg)^{M-1}\!\Big(1\!+\!\frac{x}{\eta}\beta\Big)^{-\eta-1}.\label{PDFsc}
\end{align}
\subsection{Maximal Radio Combining -- MRC}
When $\mathrm{UE}_0$ uses the MRC scheme, the error in decoding the receiving message is defined as
\begin{align}
\epsilon&=\mathbb{P}\Big(\sum_{i=1}^{M}\mathrm{SIR}_i<\theta\Big|\Phi\Big)=\mathbb{P}\big(\Psi<\theta\Big|\Phi\big)=F_{\Psi}(\theta),\label{Fpsi}
\end{align}
where $\Psi=\sum_{i=1}^{M}\mathrm{SIR}_i$. From Remark~\ref{Re1}, $\Psi$ can be represented as $\frac{\eta}{\beta}\sum_{i=1}^{M}\varphi_i$ where
%
$f_{\varphi_i}(x)=\eta(1+x)^{-\eta-1}$. 
Thus,
\begin{align}
\epsilon&\approx\mathbb{P}\Big(\sum_{i=1}^{M}\varphi_i<\frac{\beta\theta}{\eta}\Big).\label{epmrc}
\end{align}
\begin{corollary}
	Assuming the distance from $\mathrm{UE}_0$ to the serving and interfering BSs is known and $\mathrm{UE}_0$ uses MRC, the maximum number of bits to be transmitted, $k^*$, while guaranteeing the reliability constraint given by $\epsilon_{\mathrm{th}}$, is approximated by
	\begin{align}
	k^*\approx n\log_2\Big(\frac{\eta}{\beta}F_{\upsilon}^{-1}\big(\epsilon_{\mathrm{th}}\big)+1\Big), \label{kmrc}
	\end{align}
	where $\upsilon=\sum_{i=1}^{M}\varphi_i$.
\end{corollary}
\begin{proof}
	Since $\epsilon\approx\mathbb{P}\big(\upsilon<\frac{\beta\theta}{\eta}\big)=F_{\upsilon}\big(\frac{\beta\theta}{\eta}\big)$ we only require to isolate $\theta$ there while using $\theta=2^{k/n}-1$.
\end{proof}

The distribution of $\upsilon$ was already attained in \cite[CDF in Eq.(4.13)]{Zhang2.2017}. Unfortunately $F_{\upsilon}(x)$  is very difficult to evaluate, therefore, very time-consuming. In fact, it is also impossible to be evaluated for many combinations of parameter values $(M,\eta,x)$, e.g, relatively small $x$ and relatively large $M$ and/or $\eta$, for which calculation crashes with the software/hardware limitations. Following result addresses that issue.
\begin{theorem}\label{th3}
	The PDF and CDF of $\upsilon$ are approximated by
	\begin{align}
	f_{\upsilon}(x)&\!\approx\!\frac{\eta^M M^{M\!-\!1}}{(M-1)!}\Big(1\!+\!\frac{x}{M}\Big)^{-1\!-\!M\eta}\!\ln^{M-1}\Big(1\!+\!\frac{x}{M}\Big),\label{pdfG}\\
	F_{\upsilon}(x)&\!\approx 1-\frac{\Gamma\big(M,\eta M\ln(1+x/M)\big)}{(M-1)!},\label{cdfG}
	\end{align}
	where \eqref{cdfG} converges to \cite[Eq.(4.13)]{Zhang2.2017} in the left tail.
\end{theorem}
\begin{proof}
The proof will be included in an extended version of this work.
\end{proof}
\begin{figure}[t!]
	\centering
	\subfigure{\includegraphics[width=0.45\textwidth]{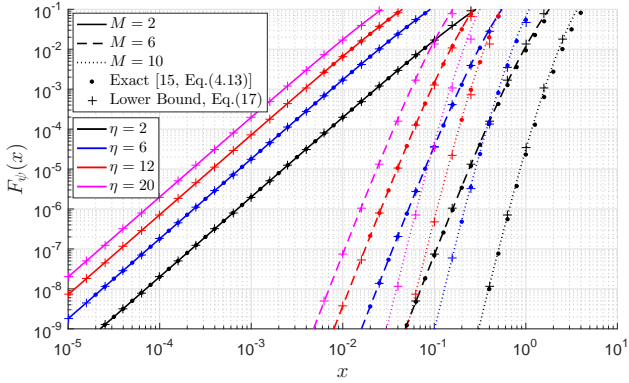}}
	\vspace*{-2mm}
	\caption{Left tail of $F_{\upsilon}(x)$. Comparison between the exact, \cite[Eq.(4.13)]{Zhang2.2017}, approximate, \eqref{cdfG}, and lower bound, \eqref{app} with $\ln(1+x/M)=x/M$, expressions.}		
	\label{Fig3}
\end{figure}
Fig.~\ref{Fig3} shows the incredible accuracy of \eqref{cdfG} in the left tail. Only a slight divergence from the exact expression is observable when $\eta$ is relatively small, e.g., $\eta=2$, at the same time that the reliability is not too restrictive, $F_{\upsilon}(x)\ge 10^{-2}$. This is in-line with the arguments we used when proving Theorem~\ref{th3}. Using expressions \eqref{pdfG} and \eqref{cdfG} is twofold advantageous: i) they are  relatively easy to evaluate and ii) they can be evaluated in regions where the exact expressions cannot. Regarding this last aspect notice that \cite[Eq.(4.13)]{Zhang2.2017} was nonviable to evaluate for $\eta=20$ and also for $\eta=12,\ M=10,\ F_{\upsilon}(x)\ge 10^{-4}$, just for mentioning two examples.

Although an easy-to-evaluate expression for $F_{\upsilon}(x)$ was given in \eqref{cdfG}, it is not analytically invertible, thus, $F_{\upsilon}^{-1}(\epsilon_{\mathrm{th}})$ requires to be computed numerically. Following result aims at alleviating this issue.
\begin{corollary}
	$F_{\upsilon}^{-1}(\epsilon_{\mathrm{th}})$ approximates to
	\begin{align}
	   \frac{(M!)^{1/M}}{\eta}\big|\ln\big(1-\epsilon_{\mathrm{th}}^{1/M}\big)\big|,\label{eqEth}
	\end{align}
	specially when $\epsilon_{\mathrm{th}}$ is very restrictive and $M$ is not too large.
\end{corollary}
\begin{proof}
	According to \cite[Eq. (8.10.11)]{Thompson.2011} we have that
	\begin{align}
	F_{\upsilon}(x)\ge \Big(1-e^{-(M!)^{-1/M}\eta M\ln\big(1+x/M\big)}\Big)^M,\label{app}
	\end{align}
	where equality holds for $M=1$ and diverges slowly when $M$ increases. Additionally, this lower bound is very tight in the left tail of the curve, e.g., when $\epsilon_{\mathrm{th}}$ is more restrictive. We require to isolate $x$ from $F_{\upsilon}(x)=\epsilon_{\mathrm{th}}$, and notice that for $\epsilon_{\mathrm{th}}\rightarrow 0$ we have $x\rightarrow 0$, thus, we can take $\ln(1+x/M)\lesssim x/M$, which makes \eqref{app} even more accurate when $\epsilon_{\mathrm{th}}$ is not too small. The tightness of the lower bound is clearly shown in Fig.~\ref{Fig3}. Finally we attain \eqref{eqEth} straightforwardly. 
\end{proof}
\section{Finite Blocklength Impact}\label{FIN}
The information theoretic analysis for infinite blocklength says that no error occurs as long as $\mathrm{SIR}>\theta$. However, if we communicate over a noisy channel and we are restricted to use a finite number of channel uses $n$, then no protocol is able to achieve perfectly reliable communication. In that sense, Polyanskiy \textit{et. al.} \cite{Polyanskiy.2010} attained, for AWGN and $n\ge 100$ channel uses, an accurate approximation that characterizes the error probability non-asymptotically as 
\begin{equation}\label{e1}
\epsilon_{_\mathrm{FB}}^{\mathrm{AWGN}}\approx Q\Biggl(\frac{C(\mathrm{SIR})-k/n}{\sqrt{V(\mathrm{SIR})/n}}\Biggl),
\end{equation}
where $C(\mathrm{SIR})=\log_2(1+\mathrm{SIR})$ is the Shannon capacity and $V(\mathrm{SIR})=\left(1-\frac{1}{(1+\mathrm{SIR})^2}\right)(\log_2e)^2$ is the channel dispersion, which measures the stochastic variability of the channel relative to a deterministic channel with the same capacity. Notice that since the CSI is available at $\mathrm{UE}_0$ the value of the $\mathrm{SIR}$ is easy to obtain, thus, the quasi-static fading channel becomes conditionally Gaussian on that, and we only require to take expectation over random variable $\mathrm{SIR}$ to attain the corresponding average error probability \cite[eq.(59)]{Yang.2014}
%
\begin{align}\label{err}
\epsilon_{_\mathrm{FB}}&\!\approx\! \mathbb{E}\Biggl[\!Q\!\Biggl(\!\frac{C(\mathrm{SIR})\!-\!k/n}{\sqrt{V(\mathrm{SIR})/n}}\!\Biggl)\!\Biggl]\!\approx\!\! \int\limits_{0}^{\infty}\!\! Q\!\Biggl(\!\frac{C(\mathrm{SIR})\!-\!k/n}{\sqrt{V(\mathrm{SIR})/n}}\!\Biggl)\!f_{_\mathrm{SIR}}(\mathrm{SIR}),
\end{align}
%
where $f_{_\mathrm{SIR}}(\mathrm{SIR})$ is given in \eqref{PDFsc} and \eqref{pdfG} for SC and MRC, respectively. Thus, taking into account the finite blocklength formulation, the maximum number of bits to be transmitted is
%
\begin{align}
&\argmax{k\in \mathbb{N}} k\nonumber\\
&\mathrm{s.t.}\ \epsilon_{_\mathrm{FB}}\le \epsilon_{\mathrm{th}},\label{prob}
\end{align}
that can be solve numerically by exhaustive search\footnote{This problem is equivalent to find the solution of $\epsilon_{_\mathrm{FB}}= \epsilon_{\mathrm{th}}$ in terms of $k$ and set $k^*\leftarrow \lfloor k\rfloor$. However, there is no numerically way to do that in one go; and solving \eqref{prob} always requires to evaluate several values of $k$ until the solution is found.}.

However, and as it has been shown in \cite{Mary.2016} for Nakagami-m and Rice channels, the quasi-static fading makes disappear the effect of the finite blocklength, thus, the asymptotic outage probability, 
which is the Laplace approximation of \eqref{err}, is a good match in those scenarios and
specially when i) $k$ is not extremely small and ii) line of sight parameter is not extremely large. The latter condition is because the fading channel tends to behave as AWGN channel when the line of sight parameter increases significantly, and at AWGN the error probability given in \eqref{e1} differs substantially from the asymptotic results.  Therefore, using the value of $k^*$ obtained in Section~\ref{INF} for the SC and MRC schemes as an initial guess when solving \eqref{prob} reduces greatly the searching time. The procedure is 
\begin{algorithmic} [1]
	\State Calculate $k^*$ according to \eqref{eqSCap}, \eqref{kmrc}, for SC and MRC schemes respectively  \label{line01} 
	\State 	$k^*\leftarrow \lfloor k^*\rfloor$
	\State Evaluate \eqref{err} \label{line02} 
	\If{$\epsilon_{_\mathrm{FB}}>\epsilon_{\mathrm{th}}$} \label{line04}    
	\State Decrease $k^*$ and return to line \ref{line02}	\label{line06}
	\EndIf        \label{line07}                     
	\State End \label{line08}         
	\vspace*{-1mm}
\end{algorithmic} 
\vspace*{-2mm}
\section{Numerical Analysis}\label{results}
Numerical results are presented in this section to evaluate the system performance in terms of maximum number of bits to be transmitted (Figs.~\ref{Fig2}, \ref{Fig4} and \ref{Fig5}) and maximum reachable rate (Fig.~\ref{Fig6}), both under stringent reliability and delay constraints. We compare SC and MRC schemes, while evaluating also the performance under the asymptotic and non-asymptotic blocklength formulations.
\begin{figure}[t!]
	\centering
	\subfigure{\includegraphics[width=0.45\textwidth]{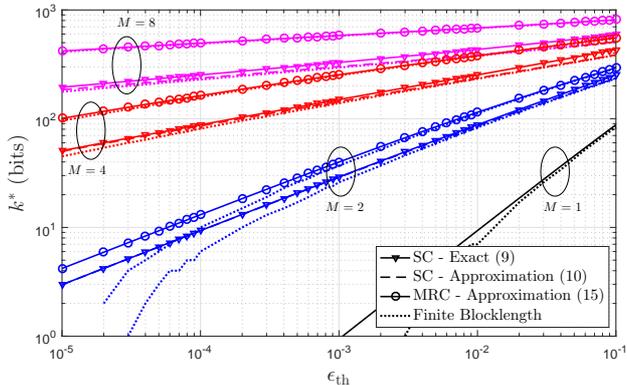}}
	\vspace*{-3mm}
	\caption{Performance of $k^*$ as function of $\epsilon_{\mathrm{th}}$ for $n=200$ channel uses, $\eta=10$, $M\in\{1,2,4,8\}$, and SC, MRC diversity schemes.  Topology with $\alpha=3.5$, $r_0=20$, $r_j=10+20j,\ j=1,...,\eta$.}
	\vspace*{-3mm}		
	\label{Fig2}
\end{figure}
Fig.~\ref{Fig2} shows the results as a function of the error probability constraint for receiving devices with $M\in\{1,2,4,8\}$ antennas and operating with a blocklength of 200 channel uses. The topology under study consists of $\eta=10$ BSs located $10+20j$ (m), $j=1,...,10$ away from the link of interest which is $20$ m long, while the path loss exponent is set to $\alpha=3.5$, thus, $\beta=0.306102$. We can notice that
\begin{itemize}
	\item operating with only one antenna is practically unfeasible for the region where $\epsilon_{\mathrm{th}}<10^{-2}$ is required, while as the number of antennas increases we can operate in the ultra-reliable region, e.g., $\epsilon_{\mathrm{th}}<10^{-3}$, with even relatively large data rates;
	\item approximation \eqref{eqSCap} is very accurate, and only when $\epsilon_{\mathrm{th}}$ and $M$ are relatively large, e.g., $M\ge 8$, the gap with respect to the exact value given in \eqref{eqSC}, although still small, can be observed. This is because \eqref{eqSCap} uses upper bound \eqref{Fsnrapp}, which is tight when the required error probability, $\epsilon_{\mathrm{th}}$, is small as discussed in Remark~\ref{Re1} and \ref{Re0}. However, for multiple antenna setups the equivalent error probability is $\epsilon_{\mathrm{th}}^{1/M}$ which increases with $M$, thus, relatively affecting the accuracy. Unfortunately, this analysis could only be done for the SC scheme since for MRC the exact expression would come from first getting the PDF of $\mathrm{SIR}_i$ from \eqref{eqSC}, which is already unfeasible. However, approximation for MRC is expected to be even more accurate since the target error probability remains unchanged;
	\item as expected, MRC overcomes the SC scheme in all the region. As the number of antennas increases, the gap increases. Interestingly and for the example topology, MRC doubles the number of bits that can be transmitted under the SC scheme when operating with $M=4$ and $M=8$ with $\epsilon_{\mathrm{th}}=10^{-5}$;
	\item the asymptotic and finite blocklength results match accurately when $k^*$ is not too small, e.g., $k^*> 30$ bits, while for extremely small data payloads the gap increases considerably. For instance, operating with $M=2$ and SC scheme, the asymptotic formulation says that we can transmit with up to $8$ bits with an error probability of $7\times 10^{-5}$, while under the finite blocklength formulation the maximum amount of information to be transmitted is reduced to only $4$ bits. 
\end{itemize}

\begin{figure}[t!]
	\centering
	\subfigure{\includegraphics[width=0.45\textwidth]{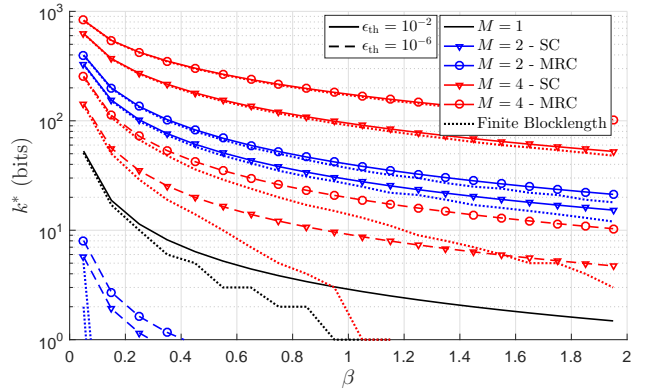}}
	\vspace*{-2mm}
	\caption{Performance of $k^*$ as function of $\beta$ for $n=200$ channel uses, $\eta=8$, $\epsilon_{\mathrm{th}}\in\{10^{-2},10^{-6}\}$, $M\in\{1,2,4\}$, and SC, MRC diversity schemes. }		
	\label{Fig4}
	\vspace*{-3mm}
\end{figure}
All the other remaining figures focus only on the results coming from evaluating the provided approximate expressions, therefore, they rely entirely on the topological parameter $\beta$. In fact, Fig.~\ref{Fig4} shows the performance as a function of $\beta$ for a setup with $8$ interfering BSs, while operating with $n=200$ channel uses. As $\beta$ increases, the performance decreases as expected from observing \eqref{eqSCap} and \eqref{kmrc}. This is because a decrement on $\beta$ is due to a larger length of the desired link and/or smaller distances to the interfering BSs and/or greater pathloss exponent. Once again we can notice that the multi-antenna configuration enables the ultra-reliability operation, while the superiority of the MRC scheme is evidenced again. Also, the asymptotic formulation is accurate when $k^*$ is not too small as previously discussed.
\begin{figure}[t!]
	\centering
	\subfigure{\includegraphics[width=0.45\textwidth]{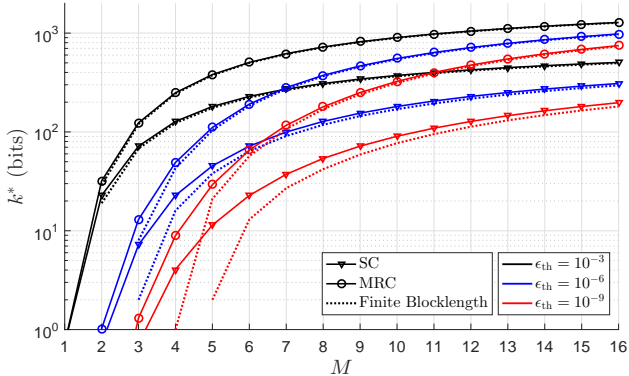}}
	\vspace*{-3mm}
	\caption{Performance of $k^*$ as function of $M$ for $n=400$ channel uses, $\eta=8$, $\beta=0.8$, $\epsilon_{\mathrm{th}}\in\{10^{-3},10^{-6},10^{-9}\}$, and SC, MRC diversity schemes. }		
	\vspace*{-2mm}
	\label{Fig5}
\end{figure}

Fig.~\ref{Fig5} shows the performance as a function of $M$ when operating with $8$ interfering BSs, $\beta=0.8$ and $n=400$ channel uses. Once again MRC outperforms SC, and notice that the gap between these two schemes tends to increase as $M$ increases. It can be observed that  the attainable data rates increase for the given reliability constraints  as $M$ increases. Notice also that for a given  $k^*$, the asymptotic formulation differs more from the finite blocklength results when the required reliability increases. As discussed before this is more obvious for the region of extremely small $k^*$, e.g., $k^*<30$ bits. 
\begin{figure}[t!]
	\centering
	\subfigure{\includegraphics[width=0.45\textwidth]{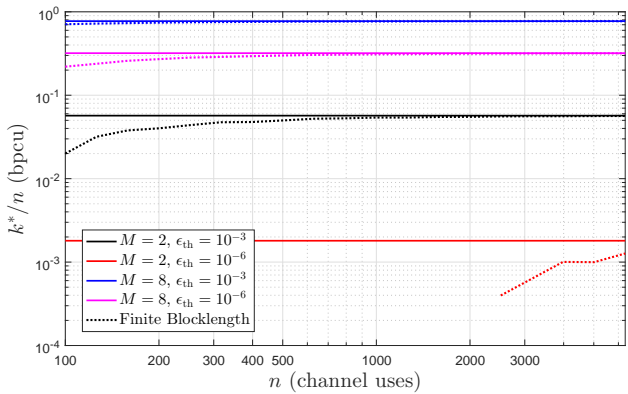}}
	\vspace*{-3mm}
	\caption{Performance of the rate, $k^*/n$, as function of blocklength, $n$, and operating with SC with $\eta=8$, $\beta=0.8$, $\epsilon_{\mathrm{th}}\in\{10^{-3},10^{-6}\}$.}		
	\vspace*{-3mm}
	\label{Fig6}
\end{figure}

Finally, Fig.~\ref{Fig6} shows the attainable rate, $k^*/n$, as a function of the blocklength $n$ and under the SC operation\footnote{Here it is only shown SC for better visualization of the results, however notice that the performance of MRC is similar but shifted up.} for a setup with $8$ interfering BSs and $\beta=0.8$. The asymptotic curves are presented as straight lines because rather on independently the $k^*$ or $n$ values, the asymptotic formulation depends on the rate $k^*/n$. Notice that the gap between the asymptotic and finite blocklength formulations tend to vanish as $n$ increases, however, this is a slower process as $k^*/n$ is smaller, which occurs when $M$ and/or $\epsilon_{\mathrm{th}}$ decrease. As shown in this and all the previous figures, the stringent the reliability requirement, the smaller the amount of information that can be transmitted.
%
\section{Conclusion}\label{conclusions}
In this paper, we proposed a rate allocation scheme for a downlink cellular system operating with stringent reliability constraints.  
The allocated rate depends on i) $\beta$, which is a function of the pathloss exponent and the distances from the served UE to all BSs; ii) the number of interfering BSs; iii) the reliability constraint; and iv) the number of antennas that are available at the UE side. We reached accurate closed-form approximations for the attainable rate when the UE operates using the SC and MRC schemes. The numerical results show the superiority of the MRC scheme and also the feasibility of the ultra-reliable operation when the number of antennas increases at the UE. Finally, we show that our analytical results remain valid even when operating at short blocklength as far as the amount of information to be transmitted is not too small. 
%
\section*{Acknowledgment}
We thank the support of Academy of Finland 6Genesis Flagship (Grant n.318927, and  n.303532, n.307492) and by the Finnish Funding Agency for Technology and Innovation (Tekes), Bittium Wireless, Keysight Technologies Finland, Kyynel, MediaTek Wireless, Nokia Solutions and Networks.

\appendices 

\bibliographystyle{IEEEtran}
\bibliography{IEEEabrv,references}
\end{document}